\documentclass[11pt,letterpaper]{article}

\usepackage[margin=1in]{geometry}
\usepackage[T1]{fontenc}
\usepackage[utf8]{inputenc}
\usepackage{lmodern,microtype}
\usepackage{amsmath,amssymb,amsthm,mathtools}
\usepackage{algorithm,algpseudocode}
\usepackage{enumitem,booktabs}
\usepackage[numbers,sort&compress]{natbib}
\usepackage[hidelinks,hypertexnames=false]{hyperref}
\usepackage{cleveref}
\usepackage{needspace,aliascnt}

\setlist{leftmargin=*,itemsep=3pt,topsep=5pt}
\newtheorem{theorem}{Theorem}[section]

\newaliascnt{lemma}{theorem}
\newtheorem{lemma}[lemma]{Lemma}
\aliascntresetthe{lemma}
\crefname{lemma}{lemma}{lemmas}
\Crefname{lemma}{Lemma}{Lemmas}
\newaliascnt{proposition}{theorem}
\newtheorem{proposition}[proposition]{Proposition}
\aliascntresetthe{proposition}
\crefname{proposition}{proposition}{propositions}
\Crefname{proposition}{Proposition}{Propositions}
\newaliascnt{corollary}{theorem}
\newtheorem{corollary}[corollary]{Corollary}
\aliascntresetthe{corollary}
\crefname{corollary}{corollary}{corollaries}
\Crefname{corollary}{Corollary}{Corollaries}
\theoremstyle{definition}
\newaliascnt{definition}{theorem}
\newtheorem{definition}[definition]{Definition}
\aliascntresetthe{definition}
\crefname{definition}{definition}{definitions}
\Crefname{definition}{Definition}{Definitions}
\theoremstyle{remark}
\newaliascnt{remark}{theorem}
\newtheorem{remark}[remark]{Remark}
\aliascntresetthe{remark}
\crefname{remark}{remark}{remarks}
\Crefname{remark}{Remark}{Remarks}
\newaliascnt{example}{theorem}

\aliascntresetthe{example}
\crefname{example}{example}{examples}
\Crefname{example}{Example}{Examples}
\newcommand{\HH}{\mathcal H}
\newcommand{\PP}{\mathcal P}
\newcommand{\QQ}{\mathbb Q}
\newcommand{\OPTu}{\operatorname{OPT}_{\mathrm U}}
\newcommand{\OPTc}{\operatorname{OPT}_{\mathrm{CVD}}}
\newcommand{\cost}{\omega}
\newcommand{\alg}{\textnormal{\textsc{UltraDelete}}}
\newcommand{\cvd}{\textnormal{\textsc{ClusterDelete}}}
\newcommand{\Ffour}{F_4}
\newcommand{\Ffive}{F_5}

\title{A $5/2$-Approximation for Weighted
Ultrametric \\ Embedding with Outliers}
\author{Chenglin Fan}
\date{}
\hypersetup{pdftitle={A 5/2-Approximation for Weighted Ultrametric Embedding with Outliers},pdfauthor={Chenglin Fan}}

\begin{document}
\maketitle

\begin{abstract}
Weighted ultrametric embedding with outliers asks for a minimum weight set of points whose deletion makes the remaining metric an ultrametric, equivalently, leaves no triple with a unique largest distance. We give a deterministic $5/2$-approximation using $O(n^5)$ arithmetic and comparison operations, improving the previous factor $3$ for arbitrary nonnegative vertex weights. Approximating the problem within any factor strictly below $2$ is UGC-hard.
\end{abstract}

\section{Introduction}
\label{sec:intro}

An ultrametric is a metric in which the largest distance in every triangle is attained at least twice; equivalently, for every distance threshold, being closer than that threshold partitions the points into clusters. We study the problem of restoring this structure by deleting points. The input is a finite metric space $(X,\rho)$, with $n=|X|$, and a nonnegative weight $w(v)$ on each $v\in X$. The goal is to find a minimum weight set $A\subseteq X$ such that $(X\setminus A,\rho)$ is an ultrametric. We call this problem \emph{Weighted Ultrametric Embedding with Outliers}. Since distances between retained points are unchanged, this is equivalently the problem of finding a maximum weight subset that embeds isometrically into an ultrametric.

Sidiropoulos et al.~\cite{sww2017} gave an $O(n^2)$ time $3$-approximation for the unweighted problem, proved NP-hardness, and showed UGC-based hardness below factor $2$. The same factor-$3$ baseline extends to arbitrary nonnegative weights by viewing violating triples, those with a unique largest distance, as three-element hitting constraints and applying local ratio~\cite{localratio2004}. The obstacle to improving this bound is that violations can occur at different scales: a threshold graph may capture one scale while hiding smaller-scale violations inside cliques. We show that two local configurations control this interaction. Removing them by local ratio leaves independent single-scale components, each reducible to Weighted Cluster Vertex Deletion.

\subsection{Results}
\label{subsec:results}

Write $w(A)=\sum_{v\in A}w(v)$, and let $\OPTu(X,w)$ denote the
minimum weight of a deletion set leaving an ultrametric. The distances
$\rho$ are understood from the instance.

\begin{theorem}[Main result]
\label{thm:main}
There is a deterministic polynomial time $5/2$-approximation algorithm for
Weighted Ultrametric Embedding with Outliers. Given an $n$-point rational
metric space $(X,\rho)$ and nonnegative rational weights $w:X\to\QQ_{\ge0}$,
it returns
$A\subseteq X$ such that $(X\setminus A,\rho)$ is an ultrametric and
\[
 w(A)\le \frac52\,\OPTu(X,w).
\]
An implementation uses $O(n^5)$ arithmetic and comparison operations and
has polynomial bit complexity in the binary input length.
\end{theorem}

The theorem improves the approximation factor from $3$ to $5/2$; the higher running time reflects our focus on the approximation ratio. The algorithm uses two local configurations of violating triples. Writing $abc$ for $\{a,b,c\}$,
\begin{equation}
\label{eq:patterns-intro}
 \Ffour=\{abc,abd,acd\},\qquad
 \Ffive=\{abc,abd,cde\}.
\end{equation}
All displayed labels are distinct. A copy of either configuration is a \emph{core}; the copy need not be induced, and an instance containing neither configuration is \emph{core-free}. In a core-free instance, intersecting violations have the same largest distance, or \emph{scale}. Hence triples linked by intersections form disjoint single-scale components. At the common scale of a component, violating triples are exactly the induced three-vertex paths in the threshold graph, so the residual problem decomposes into independent Weighted Cluster Vertex Deletion (CVD) instances. We prove this exact decomposition in \Cref{sec:structure}.

\subsection{Approach}

Each pattern in~\eqref{eq:patterns-intro} gives a local ratio certificate. For $\Ffour$, assign cost $2$ to the common vertex $a$ and cost $1$ to $b,c,d$; for $\Ffive$, use unit costs. In either case the total cost is $5$, while every set hitting the three displayed triples has cost at least $2$. Whenever a core is present, we subtract the largest multiple of its certificate that keeps all residual weights nonnegative and delete the vertices whose residual weight becomes zero. Each such reduction has ratio $5/2$, and the residual instance is core-free.

For scale separation, take intersecting violations of different scales and let $r$ be the larger scale. In the graph joining pairs at distance below $r$, the smaller-scale triple is a clique and the larger-scale triple an induced three-vertex path. A vertex adjacent to some but not all clique vertices gives an $\Ffour$; otherwise the clique vertices have identical outside neighborhoods, and replacing the unique clique vertex on the path yields an $\Ffive$. Thus different scales cannot meet in a core-free instance. We then apply the weighted CVD $2$-approximation of Aprile et al.~\cite{aprile2023} to the disjoint threshold graphs; local ratio combines this factor $2$ terminal step with the factor $5/2$ core reductions. Appendix~\ref{app:classification} shows that every configuration on at most five vertices admitting such a weighted $5/2$ certificate contains one of the two cores.

\subsection{Related work}

The closest predecessor is Sidiropoulos et al.~\cite{sww2017}, whose factor-$3$ ultrametric algorithm minimizes the number of deleted points; their UGC-based hardness below factor $2$ already holds in the unweighted case. Our weighted factor-$3$ baseline is the direct local ratio bound for the corresponding hitting-set formulation. Cluster Vertex Deletion is another three-element hitting problem: Fiorini et al.~\cite{fiorini2020} combined local ratio with graph structure, and Aprile et al.~\cite{aprile2023} later gave a tight factor-$2$ approximation for general weighted CVD. Our contribution is to reduce the residual metric instance to disjoint single-scale CVD instances. The patterns $\Ffour$ and $\Ffive$ also occur in extremal hypergraph theory~\cite{bollobas1974,keevashmubayi2004}; here they serve as local deletion certificates whose absence forces intersecting metric violations to have the same scale.

A different line of work repairs a distance matrix by editing entries rather than deleting points. Fan et al.~\cite{FanRB22} introduced Metric Violation Distance, which minimizes the number of pairwise distances changed to restore a metric. Cohen-Addad et al.~\cite{cohenaddad2025} and Charikar and Gao~\cite{charikar2024} study related metric and ultrametric fitting objectives. Other outlier-embedding work allows distortion or different target spaces, including nested embeddings~\cite{chawla2024}, probabilistic embeddings into trees~\cite{chawla2026}, and Euclidean embeddings with outliers and distance violations~\cite{bentert2025}. Our objective instead deletes weighted points while preserving every retained distance exactly.

\section{Violations and local cost certificates}
\label{sec:cores}

\subsection{Notation and input model}

The input consists of the full rational distance matrix of $(X,\rho)$ and nonnegative rational vertex weights, all encoded in binary; distinct points have positive distance. For any weight or cost function $q$ and subset $S$ of its domain, write $q(S)=\sum_{v\in S}q(v)$ and let $q|_S$ denote its restriction to $S$. For $S\subseteq X$, let $\rho|_S$ be the restriction to $S\times S$ and abbreviate $(S,\rho|_S)$ as $(S,\rho)$. Sets of at most two points are ultrametric. A three-point set $T\subseteq X$ is \emph{violating}, or \emph{bad}, if one pairwise distance is strictly larger than the other two, and its \emph{scale} is
\[
 \sigma(T)=\max_{u,v\in T}\rho(u,v).
\]
The violating triples form the $3$-uniform \emph{violation hypergraph}
\[
 \HH_X=\{T\in\tbinom X3:T\text{ is violating}\}.
\]
For $Y\subseteq X$, let $\HH_Y$ be the violation hypergraph of $(Y,\rho)$ and $\HH_X[Y]$ the hyperedges of $\HH_X$ contained in $Y$. Since deleting points does not change retained distances,
\begin{equation}
\label{eq:heredity}
 \HH_Y=\HH_X[Y].
\end{equation}
Thus deleting vertices creates neither new violations nor new cores.

A set $A\subseteq X$ is feasible exactly when it is a transversal of $\HH_X$, so
\begin{equation}
\label{eq:opt-definition}
 \OPTu(X,w)=\min\{w(A): A\subseteq X,\ A\cap T\ne\varnothing
                       \text{ for every }T\in\HH_X\}.
\end{equation}
For a $3$-uniform hypergraph $\HH$ on $W$ and nonnegative cost vector $c$, let $\tau_c(\HH)$ be its minimum transversal cost. The \emph{$2$-shadow} of $\HH$ joins two vertices whenever they occur in a common hyperedge; its connected components are the vertex components used below. A component is \emph{nontrivial} if it contains a hyperedge, equivalently if its hyperedges can be linked by a sequence of pairwise intersections; vertices in no hyperedge are isolated.

\subsection{The two cores}

\begin{definition}[Canonical cores]
\label{def:cores}
A Type-I core is a copy of $\Ffour=\{abc,abd,acd\}$ on four distinct vertices.
A Type-II core is a copy of $\Ffive=\{abc,abd,cde\}$ on five distinct vertices.
An instance is \emph{core-free} if neither occurs in its violation
hypergraph. Both types are non-induced patterns.
\end{definition}

On four vertices, any three distinct triples have a common vertex and form a
Type-I core. The three displayed edges of a Type-II core have empty common
intersection. Conversely, three distinct triples with union of size five and
empty common intersection form a Type-II core up to relabeling: the three
pairwise intersections have sizes $2,1,1$, and the two triples meeting in two
vertices determine the labeling in \Cref{def:cores}.

\begin{lemma}[Explicit cost vectors]
\label{lem:core-costs}
If a core is present on $W\subseteq X$, one can assign positive integral costs
$c$ to $W$, and zero cost to $X\setminus W$, so that
\begin{equation}
\label{eq:core-costs}
 c(X)=5,\qquad c(K)\ge2\quad\text{for every feasible deletion set }K.
\end{equation}
For Type I use $(c(a),c(b),c(c),c(d))=(2,1,1,1)$; for Type II use unit costs.
Consequently every set $A\subseteq X$ satisfies
\[
 c(A)\le\frac52\,\OPTu(X,c).
\]
\end{lemma}
\begin{proof}
For Type I, a transversal containing $a$ has cost at least $2$. A transversal
not containing $a$ must meet each of $bc,bd,cd$ and therefore contains at least
two of $b,c,d$, again at cost at least $2$. For Type II, no single vertex meets
all three displayed edges, so every transversal has at least two vertices.
Both vectors have total cost $5$. A feasible metric deletion set must hit
these displayed edges, irrespective of any additional violations. Hence
$\OPTu(X,c)\ge2$, while $c(A)\le5$ for every $A$.
\end{proof}

The vector $c$ is an auxiliary certificate, not the input weight vector.
Suppose $Y\subseteq X$ is the current vertex set and $\cost$ is a positive
residual weight function on $Y$. For a core on $W\subseteq Y$, let $c$ be
its certificate, restricted to $Y$, and set
\begin{equation}
\label{eq:scaling}
 \lambda=\min_{v\in W}\frac{\cost(v)}{c(v)},\qquad
 \cost'=\cost-\lambda c.
\end{equation}
Then $\cost'\ge0$, and at least one vertex of $W$ has zero residual cost.
Deleting such vertices is free with respect to $\cost'$, while their original
cost remains accounted for in the subtracted vector $\lambda c$, as in the
standard local ratio method~\cite{localratio2004}.

\begin{proposition}[Completeness on at most five vertices]
\label{prop:classification}
Let $\HH$ be a $3$-uniform hypergraph on $W$, where $|W|\le5$. There is a
nonzero nonnegative rational cost vector $c$ such that
\[
 c(W)\le\frac52\,\tau_c(\HH)
\]
if and only if $\HH$ contains a Type-I or Type-II core. If neither occurs,
then for every nonnegative cost vector $c$,
\[
 \tau_c(\HH)\le\frac13c(W).
\]
\end{proposition}

If a core is present, its vector from \Cref{lem:core-costs} gives the desired
certificate. For the reverse implication, Appendix~\ref{app:classification} shows
that the $2$-shadow of every core-free hypergraph on at most five vertices
has a proper three-coloring. Each color class then meets every hyperedge,
so one has cost at most $c(W)/3$.

\section{Scale separation and exact graph decomposition}
\label{sec:structure}

We now establish the structural decomposition of a core-free instance using only comparisons among pairwise distances.

\subsection{Threshold graphs}

For $r>0$ and $W\subseteq X$, let $G_r[W]$ be the graph on $W$ defined by
\begin{equation}
\label{eq:threshold}
 uv\in E(G_r[W])\quad\Longleftrightarrow\quad \rho(u,v)<r.
\end{equation}
Write $G_r=G_r[X]$, and let $P_3$ denote the path on three vertices.
An induced $P_3$ has two edges and one nonedge. In a threshold graph, it is
a violating triple: its nonedge has distance at least $r$, while its two
edges have distance below $r$. Conversely, a violating triple of scale
exactly $r$ induces a $P_3$, while a violating triple of smaller scale
induces a triangle. The \emph{closed neighborhood} of a vertex consists of the vertex itself
and all its neighbors. Two distinct vertices are \emph{true twins} if their
closed neighborhoods coincide. They are adjacent and have the same adjacency
to every other vertex, so an induced $P_3$ cannot contain two of them.
Replacing a vertex of an induced $P_3$ by a true twin outside the path
preserves the induced $P_3$.

\subsection{Intersecting violations have the same scale}

\begin{lemma}[Scale separation]
\label{lem:scale-separation}
Let $S,T\in\HH_X$ with $S\cap T\ne\varnothing$ and
$\sigma(S)\ne\sigma(T)$. The induced hypergraph $\HH_X[S\cup T]$ contains a
Type-I or Type-II core.
\end{lemma}
\begin{proof}
Interchange $S$ and $T$ if necessary so that
$\sigma(S)<\sigma(T)=r$, and write $S=\{a,b,c\}$. Put $W=S\cup T$ and
consider $G=G_r[W]$. By the threshold observations above, $G[S]$ is a clique
and $G[T]$ is an induced $P_3$. Every other induced $P_3$ found in $G$ is also
a violating triple of the metric.

Suppose first that some $z\in W\setminus S$ has a nonempty proper
neighborhood in $S$. Up to relabeling $a,b,c$, there are two possibilities.
If $z$ is adjacent only to $a$, then $zab$ and $zac$ are induced $P_3$'s.
Together with the original violating triple $abc$, they form
\[
 abc,\quad abz,\quad acz,
\]
a Type-I core with common vertex $a$. If $z$ is adjacent to $a,b$ but not to
$c$, then $zac$ and $zbc$ are induced $P_3$'s. The triples $abc,acz,bcz$
then form a Type-I core with common vertex $c$.

We may therefore assume that every vertex of $W\setminus S$ is adjacent to
all of $S$ or to none of $S$. This common outside neighborhood, together
with the fact that $S$ is a clique, makes its vertices true twins in $G$.
The induced path $G[T]$ cannot contain two of them, so $|S\cap T|=1$.
Relabel the vertices so that
\[
 S=\{a,b,s\},\qquad T=\{s,d,e\},
\]
where $a,b,s,d,e$ are distinct. Both $a$ and $b$ lie outside $T$ and are
true twins of $s$, so replacing $s$ by either vertex preserves the induced
path. Thus $dea$ and $deb$, together with the original violation $abs$,
form a Type-II core.
\end{proof}

\begin{remark}
The threshold $r$ is the larger violation scale, which may be smaller than
the diameter of $S\cup T$. Distances between the two triples may exceed $r$:
an induced $P_3$ remains violating because its nonedge distance is at least
$r$. True-twin replacement applies whether the shared vertex is an endpoint
or the middle vertex of the path.
\end{remark}

\begin{corollary}[One scale per violation component]
\label{cor:common-scale}
In a core-free instance, each nontrivial component $V_i$ of $\HH_X$ has a
number $r_i$ such that $\sigma(T)=r_i$ for every $T\in\HH_X[V_i]$.
\end{corollary}
\begin{proof}
By \Cref{lem:scale-separation}, intersecting violating triples have equal
scale. Any two edges in the same nontrivial hypergraph component are joined
by a sequence of intersecting edges. Equality of scales propagates along
this sequence.
\end{proof}

Only the violating triples in a component share the same scale. Pairwise
distances inside the component need not all be equal, and different
components may have different scales.

\subsection{Feasibility-preserving reduction to CVD}

A \emph{cluster graph} is a disjoint union of cliques. For a graph $G$,
let $\PP_3(G)$ be the family of three-vertex sets inducing a $P_3$.
A graph is a cluster graph if and only if $\PP_3(G)=\varnothing$.
Indeed, if a connected component is not a clique, the first three vertices
of a shortest path between two nonadjacent vertices form an induced $P_3$.

Given nonnegative vertex weights $u$, Weighted Cluster Vertex Deletion
(CVD) asks for a minimum weight set $B$ such that $G-B$ is a cluster graph.
Here $G-B$ is the subgraph induced by $V(G)\setminus B$. We write
\[
 \OPTc(G,u)=\min\{u(B):B\subseteq V(G),\ G-B\text{ is a cluster graph}\}.
\]

\begin{theorem}[Exact single-scale decomposition]
\label{thm:decomposition}
Let $(X,\rho)$ be core-free, with nontrivial violation components
$V_1,\ldots,V_m$ and common scales $r_1,\ldots,r_m$. Let
$G_i=G_{r_i}[V_i]$. Then
\begin{equation}
\label{eq:obstruction-equality}
 \HH_X[V_i]=\PP_3(G_i)\qquad (1\le i\le m).
\end{equation}
Consequently, for every $A\subseteq X$,
\begin{equation}
\label{eq:feasibility-equality}
 (X\setminus A,\rho)\text{ is ultrametric}
 \quad\Longleftrightarrow\quad
 G_i-(A\cap V_i)\text{ is a cluster graph for every }i.
\end{equation}
For every nonnegative weight function $w$,
\begin{equation}
\label{eq:opt-decomposition}
 \OPTu(X,w)=\sum_{i=1}^{m}\OPTc(G_i,w|_{V_i}).
\end{equation}
\end{theorem}
\begin{proof}
Every violating triple in $V_i$ has one distance equal to $r_i$ and its other
two distances strictly below $r_i$, by \Cref{cor:common-scale}. It therefore
induces a $P_3$ in $G_i$. Conversely, if a triple induces a $P_3$ in $G_i$, its
nonedge distance is at least $r_i$ and its two edge distances are smaller
than $r_i$. The nonedge is the unique largest pair, so the triple belongs to
$\HH_X[V_i]$. This proves~\eqref{eq:obstruction-equality}.

Every hyperedge of $\HH_X$ lies in one of the components $V_i$, by the
definition of a component. Deleting $A$ hits all these hyperedges exactly
when $A\cap V_i$ hits $\PP_3(G_i)$ for every $i$. This is equivalent to each
residual graph being a cluster graph and proves~\eqref{eq:feasibility-equality}.
In particular, triples using different components cause no additional
constraints: a bad triple of that kind would itself join the components.
Vertices in no bad triple may always be retained.

For the optimum equality, any feasible $A$ restricts to feasible CVD solutions
in the disjoint $V_i$, and hence
\[
 w(A)\ge\sum_i w(A\cap V_i)\ge\sum_i\OPTc(G_i,w|_{V_i}).
\]
Conversely, choose optimal CVD solutions $A_i\subseteq V_i$ and let
$A=\bigcup_i A_i$. This union is feasible by~\eqref{eq:feasibility-equality},
and its weight is exactly the right hand side of~\eqref{eq:opt-decomposition}.
\end{proof}

\begin{theorem}[Weighted CVD approximation; Aprile et al.~\cite{aprile2023}]
\label{thm:cvd}
There is a deterministic polynomial time algorithm $\cvd$ that, for a graph
$G$ and nonnegative rational vertex weights $u$, returns $B\subseteq V(G)$
such that $G-B$ is a cluster graph and
\[
 u(B)\le2\,\OPTc(G,u).
\]
The algorithm has an $O(|V(G)|^4)$ arithmetic-operation implementation.
\end{theorem}

Zero-weight vertices may be deleted before applying this theorem and added to
the returned set afterward. This does not affect the approximation guarantee.

\begin{corollary}[A $2$-approximation on core-free instances]
\label{cor:corefree-two}
Apply $\cvd$ separately to the graphs in \Cref{thm:decomposition}, retaining
all isolated vertices of the violation hypergraph. The union of its deletion
sets is feasible for the metric and has weight at most $2\OPTu(X,w)$.
\end{corollary}
\begin{proof}
Feasibility follows from~\eqref{eq:feasibility-equality}. Since the components
are disjoint, the costs add, and \Cref{thm:cvd} bounds their sum by twice the
right hand side of~\eqref{eq:opt-decomposition}.
\end{proof}

The decomposition can be built by scanning all $\binom n3$ triples. For each
violation, join its three vertices in an auxiliary graph and record its scale.
The connected components give the sets $V_i$, one stored violation gives each
$r_i$, and a final pair scan constructs $G_i$.

\section{Algorithm and approximation analysis}
\label{sec:algorithm}

\subsection{A single pass over candidate cores}

The algorithm first performs local reductions and then solves the residual core-free instance by graph decomposition. It maintains an active set $Y$, deleted set $D=X\setminus Y$, and residual costs $\cost$, initially $\cost=w$ with all zero-weight vertices placed in $D$. Fix an order on $X$ and enumerate the $O(n^5)$ injective labelings of the Type-I and Type-II templates. A candidate is used only when all its vertices are active and its three specified triples are violating; applying \Cref{lem:core-costs} and~\eqref{eq:scaling} then makes at least one residual weight zero, and all such vertices are removed. The enumeration is performed once: by~\eqref{eq:heredity} deletions create no new core, while every used candidate immediately loses a vertex, so an examined candidate cannot reappear.

\begin{algorithm}[H]
\caption{$\alg(X,\rho,w)$}
\label{alg:main}
\begin{algorithmic}[1]
\Require Rational metric $(X,\rho)$ and weights $w:X\to\QQ_{\ge0}$
\Ensure A deletion set $A$ leaving an ultrametric
\State $\cost\gets w$; $D\gets\{v\in X:w(v)=0\}$; $Y\gets X\setminus D$
\For{each injective Type-I or Type-II labeling on $X$, in a fixed order}
  \State Let $W$ be its vertex set and $T_1,T_2,T_3$ its required triples
  \If{$W\subseteq Y$ and $T_1,T_2,T_3\in\HH_X$}
    \State Let $c$ be its cost vector from \Cref{lem:core-costs}, zero outside $W$
    \State $\lambda\gets\min_{v\in W}\cost(v)/c(v)$
    \State $\cost(v)\gets\cost(v)-\lambda c(v)$ for every $v\in W$
    \State $Z\gets\{v\in W:\cost(v)=0\}$
    \State $D\gets D\cup Z$; $Y\gets Y\setminus Z$
  \EndIf
\EndFor
\State Find the nontrivial components $V_1,\ldots,V_m$ of $\HH_Y$
\For{$i=1,\ldots,m$}
  \State Let $r_i$ be the scale of any violating triple in $V_i$
  \State Construct $G_i$ on $V_i$ with $uv\in E(G_i)$ iff $\rho(u,v)<r_i$
  \State $B_i\gets\cvd(G_i,\cost|_{V_i})$
\EndFor
\State \Return $D\cup\bigcup_{i=1}^{m}B_i$
\end{algorithmic}
\end{algorithm}

All tests use the original distance matrix. Candidates can be streamed, so
there is no need to store all cores or the full violation hypergraph. Multiple
labelings of the same core affect only the constant factor. With a fixed
enumeration order and deterministic CVD calls, the algorithm is deterministic.

\subsection{Feasibility and progress}

\begin{lemma}[Reduction invariant]
\label{lem:invariant}
Throughout the local-reduction phase, $\cost$ is nonnegative, it vanishes on
$D$, and it is strictly positive on $Y$. Each used core removes at least one
active point. At the end of the enumeration, $(Y,\rho)$ is core-free.
\end{lemma}
\begin{proof}
The assertions about costs hold initially. In a used core, $c$ is positive
on $W\subseteq Y$, so the minimum in~\eqref{eq:scaling} is positive. The
update leaves all residual costs nonnegative and gives zero cost to at least
one vertex of $W$. No other point changes cost. Moving exactly those new
zero-cost vertices to $D$ maintains the invariant and removes at least one
active point.

Suppose a core survived the full enumeration. Consider any labeling of that
core encountered in the enumeration. All of its vertices survived to the
end and were therefore active at that earlier time. Its triples were
violating then, since their distances never change. The algorithm would have
used that candidate and removed one of its vertices, a contradiction.
\end{proof}

\Needspace{8\baselineskip}
\begin{lemma}[Feasibility]
\label{lem:feasibility}
The set returned by \Cref{alg:main} is a feasible ultrametric deletion set.
\end{lemma}
\begin{proof}
The remaining instance $Y$ is core-free by \Cref{lem:invariant}. The scales
chosen in the final phase are consequently well defined by
\Cref{cor:common-scale}. Each $G_i-B_i$ is a cluster graph by
\Cref{thm:cvd}, so \Cref{thm:decomposition} implies that
$Y\setminus\bigcup_iB_i$ is ultrametric. This is exactly the set left after
deleting $D\cup\bigcup_iB_i$ from $X$.
\end{proof}

\subsection{Charging the original weights}

Let $s$ be the number of used cores. For $j=1,\ldots,s$, let $c_j$ be the
local vector used in the $j$th reduction, extended by zero to $X$, and let
$\lambda_j>0$ be its scaling coefficient. Write $Y_*$ for the final active
set and $\cost_*$ for its residual costs, extended by zero to $X\setminus Y_*$.

\begin{lemma}[Cost decomposition and lower bound]
\label{lem:cost-decomposition}
The vectors produced by the algorithm satisfy
\begin{equation}
\label{eq:cost-decomposition}
 w=\cost_*+\sum_{j=1}^{s}\lambda_jc_j.
\end{equation}
Moreover,
\begin{equation}
\label{eq:lower-bound}
 \OPTu(X,w)\ge
 2\sum_{j=1}^{s}\lambda_j+\OPTu(Y_*,\cost_*|_{Y_*}).
\end{equation}
\end{lemma}
\begin{proof}
Each update subtracts exactly $\lambda_jc_j$ from the current vector. A
point is removed only when its residual cost is zero, so extending the final
vector by zero does not discard any unaccounted cost. Telescoping the updates
proves~\eqref{eq:cost-decomposition}.

Let $K$ be any feasible deletion set for the original instance. The three
triples of every used core are original violations. Thus
$K$ meets all three and $c_j(K)\ge2$ by \Cref{lem:core-costs}. Also,
$K\cap Y_*$ is feasible for the induced metric on $Y_*$, so
\[
 \cost_*(K)=\cost_*(K\cap Y_*)
 \ge\OPTu(Y_*,\cost_*|_{Y_*}).
\]
Applying~\eqref{eq:cost-decomposition} to $K$ gives
\[
 w(K)\ge2\sum_j\lambda_j+\OPTu(Y_*,\cost_*|_{Y_*}).
\]
Minimizing over $K$ proves~\eqref{eq:lower-bound}.
\end{proof}

\Needspace{14\baselineskip}
\begin{theorem}[Approximation guarantee]
\label{thm:approximation}
The set $A$ returned by \Cref{alg:main} satisfies
$w(A)\le(5/2)\OPTu(X,w)$.
\end{theorem}
\begin{proof}
Let $R=\OPTu(Y_*,\cost_*|_{Y_*})$. The points deleted during the reduction
phase have zero residual cost. By \Cref{cor:corefree-two}, the terminal phase
therefore satisfies $\cost_*(A)\le2R$. For each local vector,
$c_j(A)\le c_j(X)=5$. Consequently,
\begin{align}
 w(A)
 &=\cost_*(A)+\sum_j\lambda_jc_j(A) \notag\\
 &\le 2R+5\sum_j\lambda_j \notag\\
 &\le\frac52\left(R+2\sum_j\lambda_j\right) \notag\\
 &\le\frac52\,\OPTu(X,w), \label{eq:final-bound}
\end{align}
where the final step is \Cref{lem:cost-decomposition}.
\end{proof}

The loss from $2$ to $5/2$ comes entirely from the local core reductions;
once the residual instance is core-free, its exact decomposition into
disjoint weighted CVD instances gives a factor-$2$ approximation.

\subsection{Running time and exact arithmetic}

\begin{lemma}[Running time]
\label{lem:runtime}
\Cref{alg:main} can be implemented with $O(n^5)$ arithmetic and comparison
operations. Its bit complexity is polynomial in the binary input length.
\end{lemma}
\begin{proof}
Enumerating injective labelings of the four- and five-vertex templates takes
$O(n^5)$ iterations. Each candidate requires only a constant number of
active-set checks and distance comparisons. A successful reduction updates
at most five residual weights, using $O(1)$ arithmetic operations.
At most $n$ reductions are performed, by \Cref{lem:invariant}, and all
candidates are streamed in a single pass. To find the residual components,
scan all triples of $Y$ and record the incident pairs of each violation in a
Boolean adjacency matrix. This takes $O(n^3)$ comparisons and $O(n^2)$
auxiliary space. A second scan of these triples records one
witness scale per component, and constructing all threshold graphs takes
$O(n^2)$ further comparisons. With $n_i=|V_i|$ and $\sum_i n_i\le n$, the
CVD calls from \Cref{thm:cvd} use
\[
 O\left(\sum_i n_i^4\right)\le O(n^4)
\]
arithmetic operations. Thus the core enumeration dominates the running time.

For bit complexity, distances are only compared as original input rationals;
each comparison is performed exactly by cross-multiplication.
Let $L$ be the total binary
length of the weights and $Q$ the product of their positive denominators;
$Q$ has $O(L)$ bits. Initially every residual weight is an integer multiple
of $1/Q$. Since each nonzero coefficient of a local vector is $1$ or $2$,
after $j$ reductions every residual weight is an integer multiple of
$1/(2^jQ)$. The weights never increase and at most $n$ reductions occur,
so their numerators and denominators have $O(L+n)$ bits. The CVD calls
therefore receive rational inputs of polynomial bit length. Equivalently,
multiplying all terminal weights by $2^sQ$, where $s\le n$ is the number of
reductions, gives nonnegative integral weights of $O(L+n)$ bits without
changing approximation ratios. The polynomial time guarantee in
\Cref{thm:cvd}, together with the exact comparisons and local updates,
proves the claim.
\end{proof}

Together, \Cref{lem:feasibility}, \Cref{thm:approximation}, and
\Cref{lem:runtime} prove \Cref{thm:main}.

\bibliography{references}

@article{FanRB22,
  author       = {Chenglin Fan and
                  Benjamin Raichel and
                  Gregory Van Buskirk},
  title        = {Metric Violation Distance: Hardness and Approximation},
  journal      = {Algorithmica},
  volume       = {84},
  number       = {5},
  pages        = {1441--1465},
  year         = {2022},
}

@article{aprile2023,
author       = {Manuel Aprile and
                  Matthew Drescher and
                  Samuel Fiorini and
                  Tony Huynh},
  title        = {A tight approximation algorithm for the cluster vertex deletion problem},
  journal      = {Math. Program.},
  volume       = {197},
  number       = {2},
  pages        = {1069--1091},
  year         = {2023},
}

@article{localratio2004,
author       = {Reuven Bar{-}Yehuda and
                  Keren Bendel and
                  Ari Freund and
                  Dror Rawitz},
  title        = {Local ratio: {A} unified framework for approxmation algorithms in memoriam:
                  {Shimon Even} 1935-2004},
  journal      = {{ACM} Comput. Surv.},
  volume       = {36},
  number       = {4},
  pages        = {422--463},
  year         = {2004},
}

@inproceedings{bentert2025,
 author       = {Matthias Bentert and
                  Fedor V. Fomin and
                  Petr A. Golovach and
                  M. S. Ramanujan and
                  Saket Saurabh},
  title        = {When Distances Lie: Euclidean Embeddings in the Presence of Outliers
                  and Distance Violations},
  booktitle    = {41st International Symposium on Computational Geometry, SoCG 2025,
                  Kanazawa, Japan, June 23-27, 2025},
  series       = {LIPIcs},
  volume       = {332},
  pages        = {15:1--15:16},
  publisher    = {Schloss Dagstuhl - Leibniz-Zentrum f{\"{u}}r Informatik},
  year         = {2025},
}

@article{bollobas1974,
 author       = {B{\'{e}}la Bollob{\'{a}}s},
  title        = {Three-graphs without two triples whose symmetric difference is contained
                  in a third},
  journal      = {Discret. Math.},
  volume       = {8},
  number       = {1},
  pages        = {21--24},
  year         = {1974},
}

@inproceedings{charikar2024,
 author       = {Moses Charikar and
                  Ruiquan Gao},
  title        = {Improved Approximations for Ultrametric Violation Distance},
  booktitle    = {Proceedings of the 2024 {ACM-SIAM} Symposium on Discrete Algorithms,
                  {SODA} 2024, Alexandria, VA, USA, January 7-10, 2024},
  pages        = {1704--1737},
  publisher    = {{SIAM}},
  year         = {2024},
}

@inproceedings{chawla2024,
 author       = {Shuchi Chawla and
                  Kristin Sheridan},
  title        = {Composition of nested embeddings with an application to outlier removal},
  booktitle    = {Proceedings of the 2024 {ACM-SIAM} Symposium on Discrete Algorithms,
                  {SODA} 2024, Alexandria, VA, USA, January 7-10, 2024},
  pages        = {1641--1668},
  publisher    = {{SIAM}},
  year         = {2024},
}

@inproceedings{chawla2026,
 author       = {Shuchi Chawla and
                  Arnold Filtser and
                  Kristin Sheridan and
                  Yonatan Trachtenberg},
  title        = {Bi-Lipschitz Extensions and Outlier Embeddings into Trees},
  booktitle    = {Approximation, Randomization, and Combinatorial Optimization. Algorithms
                  and Techniques, {APPROX} , and {RANDOM} },
  series       = {LIPIcs},
  volume       = {392},
  pages        = {17:1--17:24},
  publisher    = {Schloss Dagstuhl - Leibniz-Zentrum f{\"{u}}r Informatik},
  year         = {2026},
}

@article{cohenaddad2025,
  author       = {Vincent Cohen{-}Addad and
                  Chenglin Fan and
                  Euiwoong Lee and
                  Arnaud de Mesmay},
  title        = {Fitting Metrics and Ultrametrics with Minimum Disagreements},
  journal      = {{SIAM} J. Comput.},
  volume       = {54},
  number       = {1},
  pages        = {92--133},
  year         = {2025},
}

@article{fiorini2020,
  author       = {Samuel Fiorini and
                  Gwena{\"{e}}l Joret and
                  Oliver Schaudt},
  title        = {Improved approximation algorithms for hitting 3-vertex paths},
  journal      = {Math. Program.},
  volume       = {182},
  number       = {1},
  pages        = {355--367},
  year         = {2020},
}

@article{keevashmubayi2004,
 author       = {Peter Keevash and
                  Dhruv Mubayi},
  title        = {Stability theorems for cancellative hypergraphs},
  journal      = {J. Comb. Theory {B}},
  volume       = {92},
  number       = {1},
  pages        = {163--175},
  year         = {2004},
}

@inproceedings{sww2017,
  author = {Sidiropoulos, Anastasios and Wang, Dingkang and Wang, Yusu},
  title = {Metric Embeddings with Outliers},
  booktitle = {ACM-SIAM Symposium on Discrete Algorithms (SODA)},
  pages = {670--689},
  year = {2017}
}
\bibliographystyle{plain}
\appendix
\section{Classification of weighted small-core certificates}
\label{app:classification}

This appendix proves \Cref{prop:classification} for arbitrary $3$-uniform hypergraphs on at most five vertices. Let $S$ be the $2$-shadow of $\HH$ on $W$. We show that $S$ is $K_4$-free and then use a proper three-coloring to bound the minimum transversal cost.

\begin{lemma}
\label{lem:shadow-k4-free}
If $|W|\le5$ and $\HH$ contains neither $\Ffour$ nor $\Ffive$, its
$2$-shadow is $K_4$-free.
\end{lemma}
\begin{proof}
Suppose $K=\{a,b,c,d\}$ spans a $K_4$ in the shadow. If $|W|=4$, at least
three of the four possible triples on $K$ must belong to $\HH$: two triples
cover only five of the six pairs. Three triples give a Type-I core. For $|W|=5$, let $e$ be the fifth vertex, and let $m$ be the number of
hyperedges entirely contained in $K$. If $m\ge3$, there is again a Type-I
core. If $m=2$, the two triples can be written as $abc,abd$. The shadow edge
$cd$ cannot be witnessed by an additional triple in $K$, so it forces $cde$.
Thus $abc,abd,cde$ form a Type-II core. If $m=1$, write the unique internal edge as $abc$. The shadow edges $ad$ and
$bd$ force $ade$ and $bde$, respectively. Now $ade,bde,abc$ form a Type-II
core. Finally, if $m=0$, the shadow edges $ab,ac,bc$ force $abe,ace,bce$.
These form a Type-I core with common vertex $e$. Every case is a
contradiction.
\end{proof}

\begin{lemma}
\label{lem:five-three-color}
Every $K_4$-free graph on at most five vertices has a proper three-coloring.
\end{lemma}
\begin{proof}
A graph on at most three vertices is immediate. A $K_4$-free graph on four
vertices has a nonadjacent pair; give that pair one color and the two other
vertices separate colors. For a graph on five vertices, if a vertex has
degree at most two, remove it, color the remaining four vertices, and extend
the coloring to it. Otherwise every vertex has degree at least three, so the
complement has maximum degree at most one. The complement must contain at
least two edges: if it had at most one, the original graph would contain a
$K_4$. These two edges are disjoint. Use their endpoint pairs as two color
classes and the remaining vertex as the third.
\end{proof}

\begin{proof}[Proof of \Cref{prop:classification}]
If a core occurs, use the vector of \Cref{lem:core-costs}, assigning zero
cost to vertices outside the displayed core. Every transversal of the entire
hypergraph hits its three displayed edges, so the minimum cost is at least
$2$ and the total cost is $5$.

Conversely, suppose no core occurs. By
\Cref{lem:shadow-k4-free,lem:five-three-color}, the $2$-shadow has a proper
three-coloring with classes $W_1,W_2,W_3$. Empty classes are allowed. Every
hyperedge forms a triangle in the shadow and hence uses all three colors.
Thus each $W_i$ meets every hyperedge. For any nonnegative vector $c$,
\[
 \tau_c(\HH)\le\min_{i=1,2,3}c(W_i)\le\frac13c(W).
\]
If $c\ne0$, then $c(W)>0$ and therefore
\[
 \frac52\tau_c(\HH)\le\frac56c(W)<c(W).
\]
No nonzero weighted $5/2$ certificate exists. This also covers hypergraphs
with fewer than four vertices and hypergraphs with no edges.
\end{proof}

\end{document}